\documentclass[a4paper,11pt]{article}
\usepackage[top=80pt, bottom=90pt, left=70pt, right=70pt]{geometry}
\usepackage[linesnumbered,ruled,vlined]{algorithm2e}
\SetKwInOut{Input}{input}
\SetKwInOut{Output}{output}
\usepackage{amsmath}
\usepackage{amssymb}
\usepackage{amsfonts}
\usepackage{amsthm}
\usepackage{cleveref}
\usepackage{mathtools}
\usepackage{stmaryrd}
\usepackage{setspace}
\usepackage{tabularx}
\usepackage{url}
\usepackage[subrefformat=parens]{subcaption}
\usepackage{tikz}
\usetikzlibrary{
  positioning,
  calc,
}
\usepackage{lipsum}
\usepackage[shortlabels]{enumitem}

\newtheorem{theorem}{\bfseries Theorem}[section] 
\newtheorem{lemma}[theorem]{\bfseries Lemma}

\newtheorem{corollary}[theorem]{\bfseries Corollary}
\newtheorem*{claim*}{\bfseries Claim}

\theoremstyle{definition}

\crefname{theorem}{Theorem}{Theorems}
\crefname{proposition}{Proposition}{Propositions}
\crefname{lemma}{Lemma}{Lemmas}
\crefname{exmp}{Example}{Examples}
\crefname{corollary}{Corollary}{Corollarys}
\crefname{claim}{Claim}{Claims}
\crefname{remark}{Remark}{Remarks}
\crefname{section}{Section}{Sections}

\newcommand{\size}[1]{\left| #1 \right|}
\newcommand{\bigO}[1]{\mathrm{O}(#1)}
\newcommand{\set}[1]{\left\{ #1 \right\}}
\newcommand{\inset}[2]{\left\{ #1 \;\middle|\; #2 \right\}}
\newcommand{\sequence}[1]{\langle #1 \rangle}

\newcommand{\und}[1]{#1^{\mathrm{und}}}

\newcommand{\outneighbor}[2]{N_#1^+(#2)}
\newcommand{\inneighbor}[2]{N_#1^-(#2)}
\newcommand{\neighbor}[2]{N_#1(#2)}

\newcommand{\source}[1]{{#1}^\mathtt{s}}
\newcommand{\sink}[1]{{#1}^\mathtt{t}}

\usepackage{ifthen}
\usepackage{tikz}
\usetikzlibrary{calc}
\usetikzlibrary{decorations.pathreplacing, decorations.markings}
\usetikzlibrary{arrows.meta}
\usetikzlibrary{bending}

\newcommand{\convexpath}[2]{
    [
            create hullnodes/.code={
                    \global\edef\namelist{#1}
                    \foreach [count=\counter] \nodename in \namelist {
                        \global\edef\numberofnodes{\counter}
                        \node[draw=none, fill=none] at (\nodename) [draw=none,name=hullnode\counter] {};
                    }
                    \node[draw=none, fill=none] at (hullnode\numberofnodes) [name=hullnode0,draw=none] {};
                    \pgfmathtruncatemacro\lastnumber{\numberofnodes+1}
                    \node[draw=none, fill=none] at (hullnode1) [name=hullnode\lastnumber,draw=none] {};
                },
            create hullnodes
        ]
    ($(hullnode1)!#2!-90:(hullnode0)$)
    \foreach [
        evaluate=\currentnode as \previousnode using \currentnode-1,
        evaluate=\currentnode as \nextnode using \currentnode+1
    ] \currentnode in {1,...,\numberofnodes} {
            -- ($(hullnode\currentnode)!#2!-90:(hullnode\previousnode)$)
            let \p1 = ($(hullnode\currentnode)!#2!-90:(hullnode\previousnode) - (hullnode\currentnode)$),
            \n1 = {atan2(\y1,\x1)},
            \p2 = ($(hullnode\currentnode)!#2!90:(hullnode\nextnode) - (hullnode\currentnode)$),
            \n2 = {atan2(\y2,\x2)},
            \n{delta} = {-Mod(\n1-\n2,360)}
            in
                {arc [start angle=\n1, delta angle=\n{delta}, radius=#2]}
        }
    -- cycle
}
\begin{document}
	\title{Independent set reconfiguration on directed graphs}
	\author{Takehiro Ito\thanks{Graduate School of Information Sciences, Tohoku University, Japan. Email: \texttt{takehiro@tohoku.ac.jp}} \and
	Yuni Iwamasa\thanks{Graduate School of Informatics, Kyoto University, Japan. Email: \texttt{iwamasa@i.kyoto-u.ac.jp}} \and
	Yasuaki Kobayashi\thanks{Graduate School of Information Science and Technology, Hokkaido University. Email: \texttt{koba@ist.hokudai.ac.jp}} \and
	Yu Nakahata\thanks{Division of Information Science, Nara Institute of Science and Technology, Japan. Email: \texttt{yu.nakahata@is.naist.jp}} \and
	Yota Otachi\thanks{Graduate School of Informatics, Nagoya University, Japan. Email: \texttt{otachi@nagoya-u.jp}} \and
	Masahiro Takahashi\thanks{Graduate School of Informatics, Kyoto University, Japan. Email: \texttt{takahashi.masahiro.63e@kyoto-u.jp}} \and
	Kunihiro Wasa\thanks{Toyohashi University of Technology, Japan. Email: \texttt{wasa@cs.tut.ac.jp}}
	}
	\date{\today}
	\maketitle
	
\begin{abstract}
\textsc{Directed Token Sliding} asks,
given a directed graph and two sets of pairwise nonadjacent vertices,
whether one can reach from one set to the other
by repeatedly applying a local operation that exchanges a vertex in the current set with one of its out-neighbors, while keeping the nonadjacency.
It can be seen as a reconfiguration process where a token is placed on each vertex in the current set,
and the local operation slides a token along an arc respecting its direction.
Previously, such a problem was extensively studied on undirected graphs, where the edges have no directions and thus the local operation is symmetric.
\textsc{Directed Token Sliding} is a generalization of its undirected variant since an undirected edge can be simulated by two arcs of opposite directions.

In this paper, we initiate the algorithmic study of \textsc{Directed Token Sliding}.
We first observe that the problem is PSPACE-complete even if we forbid parallel arcs in opposite directions
and that the problem on directed acyclic graphs is NP-complete and W[1]-hard parameterized by the size of the sets in consideration.
We then show our main result: a linear-time algorithm for the problem on directed graphs whose underlying undirected graphs are trees, which are called polytrees.
Such a result is also known for the undirected variant of the problem on trees~[Demaine et al.~TCS 2015],
but the techniques used here are quite different because of the asymmetric nature of the directed problem.
We present a characterization of yes-instances based on the existence of
a certain set of directed paths,
and then derive simple equivalent conditions from it by some observations,
which admits an efficient algorithm.
For the polytree case,
we also present a quadratic-time algorithm that outputs, if the input is a yes-instance, one of the shortest reconfiguration sequences.
\end{abstract}
\begin{quote}
\end{quote}

\section{Introduction}
In a reconfiguration problem, we are given an instance of a decision problem with two feasible solutions $\source{I}$ and $\sink{I}$.
The task is to decide if we can transform $\source{I}$ into $\sink{I}$ by repeatedly applying a modification rule while maintaining the feasibility of intermediate solutions.
Reconfiguration problems are studied for several problems such as \textsc{Independent Set}~\cite{HD05,KMM12,BKW14,DDFHIOOUY15,FHOU15,HU16,BB17,BMP17,LM19,BKLMOS20,IKOSUY20}, \textsc{Dominating Set}~\cite{SMN16,HIMNOST16,LMPRS18}, \textsc{Clique}~\cite{IOO15}, \textsc{Matching}~\cite{BHIM19,BBHIKMMW19,IKKKO19}, and \textsc{Graph Coloring}~\cite{BC09,OSIZ18,BHIKMMSW20}. (See also surveys~\cite{Heu13, Nis18}.)

Among the problems, \textsc{Independent Set Reconfiguration} is one of the most well-studied problems. There are three different modification rules studied in the literature: Token Addition and Removal (TAR)~\cite{IDHPSUU11,KMM12}, Token Jumping (TJ)~\cite{IKOSUY14,IKO14,BMP17}, and Token Sliding (TS)~\cite{HD05,BKW14,DDFHIOOUY15,FHOU15,HU16,BB17,LM19,BKLMOS20,IKOSUY20}.
TAR allows us to add or remove any vertex in the current independent set as long as the size of the resultant set is at least a given threshold.
TJ allows us to exchange any vertex in the set with any vertex outside the set.
TS is a restricted version of TJ that requires the exchanged vertices to be adjacent.
We call \textsc{Independent Set Reconfiguration} under TS simply \textsc{Token Sliding}.
Observe that all these rules are \emph{symmetric}: If one can transform $\source{I}$ into $\sink{I}$, then one also can transform $\sink{I}$ into $\source{I}$.
Our question is: What happens if we adopt an \emph{asymmetric} rule?

In this paper, we study a directed variant of \textsc{Token Sliding}, which we call \textsc{Directed Token Sliding}.
In this problem, the input graph is directed and we can slide tokens along the arcs only in their directions.
Here, we say that a set $I$ of vertices of a directed graph $G$ is an independent set when it is an independent set in the underlying undirected graph of $G$.\footnote{One may define an independent set of a digraph in an asymmetric way: A subset $I$ of the vertex set $V$ of the directed graph $G$ is an independent set if $u \in I$ implies $v \notin I$ for all $u, v \in V$ such that $G$ has an arc $(u, v)$. However, this definition is the same as ours. (If $u$ and $v$ are adjacent in the underlying undirected graph, at most one of $u$ and $v$ can belong to $I$.)}
If we allow two arcs with the opposite directions, this problem is a generalization of \textsc{Token Sliding}, and thus PSPACE-complete in general.

We show three results for \textsc{Directed Token Sliding}.
First, we show that the problem is PSPACE-complete even on oriented graphs, where antiparallel arcs (two parallel arcs with opposite directions) are not allowed.
Second, if we restrict input graphs to directed acyclic graphs (DAGs), we prove that the problem is NP-complete.
Moreover, the problem is $W[1]$-hard parameterized by the number of tokens.
As for our positive and main result, we show that the problem can be solved in linear time on polytrees, that is, digraphs whose underlying graphs are trees.
For a polytree, we show that all reconfiguration sequences have the same length.
We can also construct one of them in $\bigO{|V|^2}$ time, where $V$ is the set of vertices of the input graph.
Note that our algorithm is optimal in the following sense:
We can show that there exists an infinite family of instances on directed paths whose reconfiguration sequences have length $\Omega(|V|^2)$ (which is easily obtained from lower bound examples for undirected paths given in \cite{DDFHIOOUY15}).
For \textsc{Token Sliding} on undirected trees,
Demaine et al.~\cite{DDFHIOOUY15} showed a linear-time algorithm to check the reconfigurability. 
They also showed an $\bigO{|V|^2}$-time algorithm to construct a reconfiguration sequence of length $\bigO{|V|^2}$ for yes-instances.
However, the output sequence is not guaranteed to be shortest.
The best known algorithm to output a shortest reconfiguration sequence for undirected trees runs in $\bigO{|V|^5}$ time~\cite{Sug18}.
In contrast to the undirected counterpart, our algorithm to construct a (shortest) reconfiguration sequence on polytrees runs in $\bigO{|V|^2}$ time, which is optimal.
It should be mentioned that our quadratic-time algorithm does not imply a quadratic-time algorithm for finding a shortest reconfiguration sequence for undirected trees since we do not allow polytrees to have antiparallel arcs.


\paragraph*{Related work}
\textsc{Token Sliding} on undirected graphs was introduced by Hearn and Demaine~\cite{HD05}.
They show that the problem is PSPACE-complete even on planar graphs with maximum degree~3.
The problem is also PSPACE-complete on bipartite graphs~\cite{LM19} and split graphs~\cite{BKLMOS20}.
In contrast to these hardness results, the problem is polynomial-time solvable on cographs~\cite{KMM12}, claw-free graphs~\cite{BKW14}, trees~\cite{DDFHIOOUY15}, cactus graphs~\cite{HU16}, interval graphs~\cite{BB17}, bipartite permutation graphs, and bipartite distance-hereditary graphs~\cite{FHOU15}.

Interestingly, \textsc{Independent Set Reconfiguration} on bipartite graphs is NP-complete under TAR and TJ~\cite{LM19}.
Such graph classes are not known for TS.
In this paper, we show that \textsc{Directed Token Sliding} is NP-complete on directed acyclic graphs (DAGs).

\textsc{Token Sliding} is also studied from the viewpoint of fixed-parameter tractability (FPT).
When the parameter is the number of tokens,
the problem is W[1]-hard both on $C_4$-free graphs and bipartite graphs, while it becomes FPT on their intersection, $C_4$-free bipartite graphs~\cite{BBDLM20}.
In this paper, we show that \textsc{Directed Token Sliding} is W[1]-hard on DAGs.

Although \textsc{Token Sliding}, as far as we know, has not been studied for digraphs, there are some studies on reconfiguration considering directions of edges.
An orientation of a simple undirected graph is an assignment of directions to the edges of the graph.
There are several studies for reconfiguring orientations of a graph, such as strong orientations~\cite{GZ83,FPS01,TIKKKMNOO22}, acyclic orientations~\cite{FPS01}, nondeterministic constraint logic~\cite{HD05}, and $\alpha$-orientations~\cite{ACHKMSV20}.
Very recently, Ito et al.~\cite{ito2021reconfiguring} studied reconfiguration of several subgraphs in directed graphs, such as directed trees, directed paths, and directed acyclic subgraphs.
Although these reconfiguration problems are defined on directed graphs, the reconfiguration rules are symmetric, meaning that if one solution $X$ can be obtained from another solution $Y$ by applying some reconfiguration rule, $Y$ can be obtained from $X$ by the same rule as well.

\section{Preliminaries}

In this section, we introduce notations on graphs and define our problem.
For a positive integer $k$, we define $[k] \coloneqq \set{1, \dots, k}$. We also define $[0] \coloneqq \emptyset$.
\subsection{Graph notation}
Let $G = (V, A)$ be a directed graph (digraph).
We denote by $V(G)$ and $A(G)$ the vertex and arc sets of $G$, respectively. 
Let $e = (u, v) \in A$.
The vertex $u$ (resp.\ $v$) is called the \emph{tail} (resp.\ \emph{head}) of $e$.
For a digraph $G$, its \emph{underlying (undirected) graph} is an undirected graph obtained by ignoring the directions of arcs in $G$.
We denote by $\und{G}$ the underlying graph of $G$.
For an undirected graph $G'$, we denote by $V(G')$ and $E(G')$ the vertex and edge sets of $G'$, respectively.
For an undirected graph $G'$ and $S \subseteq V(G')$, $S$ is an \emph{independent set in $G'$} if, for all two different vertices $u, v \in S$, we have $\{u, v\} \notin E(G')$.
For a digraph $G$ and $S \subseteq V(G)$, $S$ is an \emph{independent set in $G$} if $S$ is an independent set in $\und{G}$.
For a digraph $G$ and $u, v \in V(G)$, $u$ and $v$ are \emph{adjacent} if they are adjacent in $\und{G}$.
For $e = (u, v) \in A(G)$, the \emph{endpoints of $e$} are $u$ and $v$.
For $e \in A(G)$ and $v \in V(G)$, $e$ is \emph{incident to $v$} if $v$ is an endpoint of $e$.
If $(u, v) \in A(G)$, $u$ is an \emph{in-neighbor of $v$ in $G$} and $v$ is an \emph{out-neighbor of $u$ in $G$}.
In addition, $u$ is a \emph{neighbor} of $v$ and $v$ is a \emph{neighbor} of $u$.
For a vertex $v$, we denote by $\outneighbor{G}{v}$ and $\inneighbor{G}{v}$ the sets of out-neighbors and in-neighbors of $v$ in $G$, respectively, that is, $\outneighbor{G}{v} \coloneqq \inset{u \in V(G)}{(v, u) \in A(G)}$ and $\inneighbor{G}{v} \coloneqq \inset{u \in V(G)}{(u, v) \in A(G)}$.
We denote by $\neighbor{G}{v}$ the set of neighbors of $v$, that is, $\neighbor{G}{v} \coloneqq \outneighbor{G}{v} \cup \inneighbor{G}{v}$.
In addition, we define $N_G[v] \coloneqq N_G(v) \cup \{v\}$.
We also define $\Gamma^+_G(v) \coloneqq \inset{(v, u)}{u \in \outneighbor{G}{v}}$, $\Gamma^-_G(v) \coloneqq \inset{(u, v)}{u \in \inneighbor{G}{v}}$, and $\Gamma_G(v) \coloneqq \Gamma^+_G(v) \cup \Gamma^-_G(v)$.
If no confusion arises, we omit the subscript $G$ from $\neighbor{G}{v}$, $\outneighbor{G}{v}$, $\inneighbor{G}{v}$, $N_G[v]$, $\Gamma^+_G(v)$, $\Gamma^-_G(v)$, and $\Gamma_G(v)$.
A \emph{directed path} $P$ is a sequence of vertices and arcs $(v_1, e_1, \dots, e_{\ell - 1}, v_\ell)$ such that, all the vertices are distinct and, for every $i \in [\ell-1]$, $e_i = (v_i, v_{i+1})$ holds.
We call $v_1$ and $v_{\ell}$ the \emph{source} and \emph{sink} of $P$, and denote by $s(P)$ and $t(P)$, respectively.
This $P$ is called a \emph{$v_1$-$v_\ell$ (directed) path} or a \emph{(directed) path from $v_1$ to $v_{\ell}$}.
The \emph{length} of $P$ is the number of arcs and we denote it by $|P|$, that is, $|P| = \ell - 1$.
For a directed path $P$ with $|P| \ge 2$, we define $s'(P) \coloneqq v_2$ and $t'(P) \coloneqq v_{\ell - 1}$.
For vertex sets $X$ and $Y$ with the same size $k$ on a digraph, we refer to a set $\mathcal{P} = \set{P_1, \dots, P_k}$ of directed paths as a \emph{directed path matching from $X$ to $Y$} if $\mathcal{P}$ have distinct sources and distinct sinks, that is, $\inset{s(P_i)}{i \in [k]} = X$ and $\inset{t(P_i)}{i \in [k]} = Y$.
Note that two sets $X$ and $Y$ may intersect in this definition.

A digraph $G$ is an \emph{oriented graph} if, for all vertices $u, v \in V(G)$, $G$ contains at most one of the possible arcs $(u, v)$ or $(v, u)$.
If $G$ contains no directed cycles, $G$ is \emph{acyclic} and such digraphs are called \emph{directed acyclic graphs (DAGs)}.
If $\und{G}$ is a tree, $G$ is a \emph{polytree}.
Similarly, if $\und{G}$ is a forest, $G$ is a \emph{polyforest}.

\subsection{Definition of \textsc{Directed Token Sliding}}
Let $\source{I}$ and $\sink{I}$ be independent sets in a digraph $G$ with $\size{\source{I}} = \size{\sink{I}}$.
A sequence $\sequence{I_0, \dots, I_\ell}$ of independent sets of $G$ is a \emph{reconfiguration sequence from $\source{I}$ to $\sink{I}$ in $G$} if $I_0 = \source{I}, I_{\ell} = \sink{I}$, and for all $i \in [\ell]$ we have $I_{i-1} \setminus I_i = \{u\}$ and $I_{i} \setminus I_{i - 1} = \{v\}$ with $(u, v) \in A$.
The \emph{length} of $\sequence{I_0, \dots, I_\ell}$ is $\ell$.
We say that \emph{$\source{I}$ is reconfigurable into $\sink{I}$ in $G$} and \emph{$\sink{I}$ is reconfigurable from $\source{I}$ in $G$} if there is a reconfiguration sequence from $\source{I}$ to $\sink{I}$.

We study the following problem:

\begin{description}
    \item[Problem] \textsc{Directed Token Sliding}
    \item[Instance] A triple $(G, \source{I}, \sink{I})$, where $G$ is a digraph and $\source{I}, \sink{I} \subseteq V(G)$ are independent sets in $G$ with $\size{\source{I}} = \size{\sink{I}}$.
    \item[Question] Is $\source{I}$ reconfigurable into $\sink{I}$?
\end{description}

In \textsc{Directed Token Sliding}, the sets $\source{I}$ and $\sink{I}$ can be seen as the initial and target positions
of ``tokens'' placed on the vertices in the sets, and then 
the problem asks whether we can move the tokens from $\source{I}$ to $\sink{I}$ by repeatedly sliding tokens along arcs keeping that the tokens are not adjacent.
The difference between our problem and \textsc{Token Sliding} is that, in the former, the input graph is directed and we can slide tokens along the arcs only in their directions.
However, since we can simulate an undirected edge with the two arcs with the opposite directions, the problem is a generalization of \textsc{Token Sliding}.
It immediately follows from the PSPACE-completeness of \textsc{Token Sliding}~\cite{HD05}
that \textsc{Directed Token Sliding} is PSPACE-complete in general.

\section{Hardness results}\label{sec:hardness}
In this section, we provide hardness results for \textsc{Directed Token Sliding}.
The first hardness result is the PSPACE-completeness of \textsc{Directed Token Sliding} for oriented graphs.
This follows from a reduction from \textsc{Token Sliding} on undirected graphs, which is PSPACE-complete~\cite{HD05}.
\begin{theorem}\label{thm:pspace-comp}
    \textsc{Directed Token Sliding} is PSPACE-complete on oriented graphs. 
\end{theorem}
\begin{proof}
    First, we show that \textsc{Directed Token Sliding} is in PSPACE.
    By a standard argument in reconfiguration problems, the problem belongs to PSPACE: By non-deterministically guessing the ``next solution'' in a reconfiguration sequence, the problem can be solved in non-deterministic polynomial space, while by Savitch's theorem~\cite{Sav70}, we can solve the problem in deterministic polynomial space as well.
    
    Next, we give a polynomial-time reduction from \textsc{Token Sliding} to \textsc{Directed Token Sliding} on oriented graphs.
    Let $(G, \source{I}, \sink{I})$ be an instance of \textsc{Token Sliding}.
    Then, we construct an instance $(G', \source{J}, \sink{J})$ of \textsc{Directed Token Sliding} such that $G'$ is an oriented graph as follows.
    For each vertex $v \in V(G)$, we make two vertices $v_1, v_2$ in $V(G')$.
    We add to $G'$ an arc $(v_1, v_2)$ for each $v \in V(G)$.
    For each edge $\{u, v\} \in E(G)$, we add arcs $(u_1, v_1), (v_1, u_2), (u_2, v_2)$, and $(v_2, u_1)$ to $G'$. 
    The resultant graph $G'$ is an oriented graph.
    We set $\source{J} = \inset{v_1}{v \in \source{I}}$ and $\sink{J} = \inset{v_1}{v \in \sink{I}}$.
    
    We show that $(G', \source{J}, \sink{J})$ is a yes-instance for \textsc{Directed Token Sliding} if and only if $(G, \source{I}, \sink{I})$ is a yes-instance for \textsc{Token Sliding}.
    We first show the if direction.
    Suppose that $(G, \source{I}, \sink{I})$ is a yes-instance for \textsc{Token Sliding}.
    Let $\sequence{I_0, \dots, I_{\ell}}$ be a reconfiguration sequence, where $I_0 = \source{I}$ and $I_{\ell} = \sink{I}$.
    We define $J_i = \inset{v_1}{v \in I_i}$.
    Note that $\source{J} = J_0$ and $\sink{J} = J_\ell$.
    We can reconfigure $J_i$ into $J_{i+1}$ as follows.
    Let $I_i \setminus I_{i+1} = \{u\}$ and $I_{i+1} \setminus I_i = \{v\}$.
    If $(u_1, v_1) \in A(G')$, by sliding the token on $u_1$ to $v_1$ along the arc, we obtain $J_{i+1}$.
    If $(u_1, v_1) \notin A(G')$, there are three arcs $(u_1, v_2), (v_2, u_2)$, and $(u_2, v_1)$ in $A(G')$. We can slide the token on $u_1$ to $v_1$ along these arcs and obtain $J_{i+1}$.
    Since $N_{G'}(u_1) = N_{G'}(u_2)$ and $N_{G'}(v_1) = N_{G'}(v_2)$, all intermediate sets are independent sets.
    By applying these operations inductively, we can construct a reconfiguration sequence from $\source{J} = J_0$ to $\sink{J} = J_{\ell}$.
    
    Conversely, we show the only-if direction.
    If $(G', \source{J}, \sink{J})$ is a yes-instance for \textsc{Directed Token Sliding}, 
    consider an arbitrary reconfiguration sequence from $\source{J}$ to $\sink{J}$ in $G'$.
    We can construct a reconfiguration sequence from $\source{I}$ to $\sink{I}$ in $G$ as follows.
    If a token on $v_1$ slides to $v_2$ along the arc $(v_1, v_2)$ in $G'$, we do not move the tokens on $G$.
    Otherwise, if a token in $G'$ moved from $u_i$ to $v_j$ for some $u, v \in V(G)$ and $i, j \in [2]$, and then we move the token on $u$ to $v$.
    Since two vertices $u$ and $v$ are adjacent in $G$ if and only if $u_i$ and $v_j$ are adjacent in $\und{G'}$ for $1 \le i,j\le2$, implying that the intermediate sets in $G$ are independent sets.
\end{proof}

The second hardness result is the NP-completeness and W[1]-hardness of \textsc{Directed Token Sliding} for DAGs.

\begin{theorem}\label{thm:dag}
    \textsc{Directed Token Sliding} on DAGs is NP-complete and $W[1]$-hard parameterized by the number of tokens.
\end{theorem}
\begin{proof}
    First, we show that \textsc{Directed Token Sliding} on DAGs is in NP.
    For every yes-instance $(G, \source{I}, \sink{I})$, each token visits a vertex at most once in every reconfiguration sequence from $\source{I}$ to $\sink{I}$ as $G$ is acyclic.
    Thus, the length of any reconfiguration sequence is $\bigO{|V(G)|^2}$, implying that the problem belongs to NP.
    
    To prove the NP-hardness, we perform a reduction from \textsc{Multicolored Independent Set} to \textsc{Directed Token Sliding}.
    In \textsc{Multicolored Independent Set}, given a graph $G$, an integer $k$, and a partition $\set{V_1, V_2, \ldots, V_k}$ of $V(G)$, we are asked to determine whether $G$ has a \emph{multicolored independent set}, that is, an independent set $X \subseteq V(G)$ of $G$ that contains exactly one vertex from each $V_i$.
    This problem is known to be not only NP-complete but also W$[1]$-hard parameterized by $k$~\cite{CFKLMPPS15}.
    
    Given an instance $(G, k, \set{V_1, \dots, V_k})$ of \textsc{Multicolored Independent Set}, we construct an instance $(G', \source{I}, \sink{I})$ for \textsc{Directed Token Sliding} such that $G'$ is acyclic and $\size{\source{I}} = \size{\sink{I}} = k + 1$.
    We define $V(G') = U \cup V(G) \cup W$, where $\size{U} = \size{W} = k + 1$.
    Let $U = \set{u_1, \dots, u_{k+1}}$ and $W = \set{w_1, \dots, w_{k+1}}$.
    We set $\source{I} = U$ and $\sink{I} = W$.
    We define $A(G') = A_{UW} \cup A_{UV} \cup A_{VW} \cup A_V$.
    Here, $A_{UW} = \inset{(u, w)}{u \in U, w \in W}$, $A_{UV} = \inset{(u_i, v)}{i \in [k], v \in V_i}$, and $A_{VW} = \inset{(v, w_i)}{i \in [k], v \in V_i}$.
    $A_V$ is the set of arcs obtained by orienting edges in $E(G)$ so that 1) the digraph $(V(G), A_V)$ is acyclic and, 2) for every $1 \le i < j \le k$, all the edges between $v \in V_i$ and $v' \in V_j$ must be directed from $v$ to $v'$.
    There always exists such an orientation by considering an arbitrary vertex ordering in each $V_i$ and orienting edges in $G[V_i]$ according to this vertex ordering. 
    In addition, it is easy to see that the obtained $G'$ is acyclic.
    
    Now we show that $(G', \source{I}, \sink{I})$ is a yes-instance for \textsc{Directed Token Sliding} if and only if $(G, k, \set{V_1, \dots, V_k})$ is a yes-instance for \textsc{Multicolored Independent Set}.
    Let $X = \set{v_1, \dots, v_k}$ be a multicolored independent set in $G$ such that $v_i \in V_i$ for $i \in [k]$.
    For each $i \in [k]$, we slide the token on $u_i$ to $v_i$ along the arc $(u_i, v_i) \in A_{UV}$.
    Since $X$ is an independent set of $G$ and $u_i$ is not adjacent to any vertex in $V_j$ with $j \neq i$, every intermediate position of tokens forms an independent set of $G$. 
    We then slide the token on $u_{k+1}$ to $w_{k+1}$ along arc $(u_{k+1}, w_{k+1})$.
    This move is valid since, for any $v \in V$, there are no arcs between $u_{k+1}$ and $v$ and between $w_{k+1}$ and $v$.
    Finally, for each $i \in [k]$, we slide the token on $v_i$ to $w_i$.
    These moves are also valid since, for $i \in [k]$ and $j \in [k+1]$ with $i \neq j$, there are no arcs whose endpoints are $v_i$ and $w_j$.
    Thus, there is a reconfiguration sequence from $\source{I}$ to $\sink{I}$.
    
    Conversely, suppose that there is a reconfiguration sequence from $\source{I}$ to $\sink{I}$ in $G'$.
    In any independent set contained in the reconfiguration sequence, if there is a token on $u_1, \dots, u_k$, the token on $u_{k+1}$ cannot move as $U$ and $W$ are completely joined by $A_{UW}$.
    In addition, if a token exists on $u_{k+1}$, none of the tokens on $u_1, \dots, u_k$ can move to the vertices in $W$ because $u_{k+1}$ is adjacent to all the vertices in $W$.
    Therefore, we must push the tokens on $u_1, \dots, u_k$ into $V(G)$ before moving the token on $u_{k+1}$.
    At this time, for each $i \in [k]$, the token on $u_i$ can move only to a vertex in $V_i$, and thus the tokens that are moved from $u_1, \dots, u_k$ to $V(G)$ must form a multicolored independent set of size $k$.
\end{proof}

\section{Linear-time algorithm for polytrees}\label{sec:polytree}
This section is devoted to proving our main result \cref{thm:polytree} below,
a linear-time algorithm for \textsc{Directed Token Sliding} on polytrees.

\begin{theorem}\label{thm:polytree}
    Let $T = (V, A)$ be a polytree. \textsc{Directed Token Sliding} on $T$ can be solved in $\bigO{|V|}$ time.
    Moreover, if the answer is affirmative,
    then all reconfiguration sequences have the same length,
    and
    one of them can be constructed in $\bigO{|V|^2}$ time.
\end{theorem}

\subsection{Directed path matching}
Let $\source{I}$ and $\sink{I}$ be independent sets in $T$ with $|\source{I}| = |\sink{I}|$ such that $\source{I}$ is reconfigurable into $\sink{I}$.
In a reconfiguration sequence from $\source{I}$ to $\sink{I}$, the $i$-th token moves along some directed path $P_i$.
Clearly $\mathcal{P} \coloneqq \set{P_1, \dots, P_k}$ forms a directed path matching from $\source{I}$ to $\sink{I}$.
Thus, the existence of a directed path matching is a trivial necessary condition for yes-instances.
However, the converse is not true in general.
Let us consider a digraph such that its underlying graph is the star with four leaves and its center has two in-neighbors and two out-neighbors.
We set $\source{I}$ (resp. $\sink{I}$) to be the set of in-neighbors (resp. out-neighbors) of the center.
Then this graph has a directed path matching from $\source{I}$ to $\sink{I}$, while it is a no-instance.
Nevertheless, directed path matchings still give us some insight on \textsc{Directed Token Sliding} for polytrees, which is vital for our linear-time algorithm.
To this end, in this subsection, we give a characterization of the existence of a directed path matching between given two sets of vertices in a polytree.
This characterization is described in terms of an invariant associated with each arc $e$.

Let $X$ and $Y$ be (not necessarily disjoint) sets of vertices of a polytree $T = (V, A)$ with $|X| = |Y|$,
and $\pi$ a bijection from $X$ to $Y$.
Since $T$ is a polytree,
for each $x \in X$ an (undirected) $x$-$\pi(x)$ path $P_x$ is uniquely determined in $\und{T}$.
For each $e \in A$,
we define $w(e; X, Y, \pi) \coloneqq |\{ x \in X \mid \overrightarrow{P_x} \text{ has } e \}| - |\{ x \in X \mid \overrightarrow{P_x} \text{ has the reverse of } e \}|$,
where $\overrightarrow{P_x}$ is a directed path obtained from $P_x$ by orienting arcs from $x$ to $\pi(x)$.
The following lemma states that $w(e; X, Y, \pi)$ does not depend on a particular $\pi$.
Let $T'$ be the polyforest obtained from $T$ by removing an arc $e$.
Let $C_e^+$ (resp.\ $C_e^-$) denote the vertices of the (weakly) connected component in $T'$ containing the head of $e$ (resp.\ the tail of $e$).
\begin{lemma}\label{lem:we}
    Let $X$ and $Y$ be (not necessarily disjoint) sets of vertices of $T$ with $|X| = |Y|$,
    and $e \in A$ an arc of $T$.
    For each bijection $\pi : X \to Y$,
    we have $w(e; X, Y, \pi) = |C_e^- \cap X| - |C_e^- \cap Y|$.
\end{lemma}
\begin{proof}
    For each $x \in X$,
    the $x$-$\pi(x)$ path $\overrightarrow{P_x}$ uses $e$ if and only if $x \in C_e^-$ and $\pi(x) \in C_e^+$,
    and $\overrightarrow{P_x}$ uses the reverse of $e$ if and only if $x \in C_e^+$ and $\pi(x) \in C_e^-$.
    Therefore we have
    \begin{align*}
        w(e; X, Y, \pi) &= |\{ x \in X \mid x \in C_e^- \text{ and } \pi(x) \in C_e^+ \}| - |\{ x \in X \mid x \in C_e^+ \text{ and } \pi(x) \in C_e^- \}|\\
        &= |C_e^- \cap X| - |C_e^- \cap Y|.\qedhere
    \end{align*}
\end{proof}

Based on this fact, we define a function $w$ on $A$ with respect to two vertex sets $X$ and $Y$:
\[
     w(e; X, Y) \coloneqq |C_e^- \cap X| - |C_e^- \cap Y| \qquad (e \in A).
\]
Then, we show the necessary and sufficient conditions mentioned above.
\begin{lemma}\label{lem:chara_matching}
    There exists a directed path matching from $X$ to $Y$ if and only if
    $w(e; X, Y) \geq 0$ for every arc $e \in A$.
\end{lemma}
\begin{proof}
    Let $e \in A$ be arbitrary.
    Suppose that  $|C^-_e \cap X| < |C^-_e \cap Y|$.
    Then, a directed path matching has at least one directed path from a vertex in $C^+_e \cap X$ to a vertex in $C^-_e \cap Y$, contradicting the direction of $e$.
    
    Conversely, suppose that $w(e; X, Y) \ge 0$ for every $e \in A$.
    Then, we claim that there is a directed path matching from $X$ to $Y$.
    We show the claim by induction on $\sum_{e \in A} w(e; X, Y)$.
    To facilitate our proof, we slightly abuse the notation $w$ and directed path matchings.
    We naturally extend the definition of $w(e; X, Y)$ for multisets $X$ and $Y$, where $|C^-_e \cap X|$ (resp.\ $|C^-_e \cap Y|$) is defined as the sum of the multiplicity of each distinct element in $C^-_e \cap X$ (resp.\ $C^-_e \cap Y$).
    For multisets $X = \{x_1, x_2, \ldots, x_k\}$ and $Y = \{y_1, y_2, \ldots, y_k\}$, a set of paths $\set{P_1, \dots, P_k}$ is a directed path matching if there exists a bijection $\pi$ from $[k]$ to $[k]$ and $P_i$ is a directed path from $x_i$ to $y_{\pi(i)}$.
    
    Suppose that $\sum_{e\in A} w(e; X, Y) = 0$.
    Since $w(e; X, Y) \ge 0$ for each $e \in A$, we have $X = Y$ and hence the claim holds.
    Suppose that $\sum_{e\in A} w(e; X, Y) > 0$.
    Since $w(e; X, Y) \ge 0$ for $e \in A$, there is an arc $e^* \in A$ with $w(e^*; X, Y) > 0$.
    Let $T'$ be the (weakly) connected component induced by arcs $e'$ with $w(e'; X, Y) > 0$ that contains $e^*$.
    $|V(T') \cap X| = |V(T') \cap Y|$ holds, as otherwise there is an arc $e'$ incoming to $T'$ or outgoing from $T'$ with $w(e'; X, Y) > 0$, contradicting the definition of $T'$.
    Choose a vertex $v$ with $N^-_{T'}(v) = \emptyset$. 
    Since the sum of in-degrees in $T'$ is $|V(T')| - 1$, such a vertex exists.
    Since $|V(T')| \geq 2$ by the existence of $e^*$, $N^+_{T'}(v)$ is nonempty.
    If $v \notin X$, there exists some vertex $u \in N^+_{T'}(v)$ such that $|C^+_{(v, u)} \cap X| \geq |C^+_{(v, u)} \cap Y|$ as otherwise $|V(T') \cap X| < |V(T') \cap Y|$.
    This implies $w((v, u); X, Y) \leq 0$, contradicting the assumption on $T'$.
    Thus, we have $v \in X$.
    Let $u \in N^+_{T'}(v)$.
    Then $w(e'; X \setminus \{v\} \cup \{u\}, Y) = w(e'; X, Y)$ for each $e' \in A \setminus \{(v, u)\}$ and $w((v, u); X \setminus \{v\} \cup \{u\}, Y) = w((v, u); X, Y) - 1 \ge 0$.
    (Note that for a multiset $X$, $X \setminus \{v\}$ means that we decrease the multiplicity of $v$ in $X$ by $1$.)
    By induction hypothesis, there is a directed path matching from $X \setminus \set{v} \cup \set{u}$ to $Y$, say $\mathcal{P'}$. 
    Then appending $v$ before the source of $P' \in \mathcal{P}'$ with $s(P') = u$ yields a directed path matching from $X$ to $Y$.
\end{proof}

By \cref{lem:we,lem:chara_matching}, we can observe the following corollaries.
\begin{corollary}\label{cor:length}
    For a yes-instance, the number of tokens passing through an arc $e$ is equal to $w(e; \source{I}, \sink{I})$ in every reconfiguration sequence.
    In particular, all reconfiguration sequences have the same length $\sum_{e \in A} w(e; \source{I}, \sink{I})$.
\end{corollary}
\begin{corollary}\label{cor:negative}
    If there exists an arc $e \in A$ such that $w(e; \source{I}, \sink{I}) < 0$, then the instance is not reconfigurable.
\end{corollary}

In the following argument, we assume $w(e; \source{I}, \sink{I}) \geq 0$ for all $e \in A$.
By depth-first search on $T$, we can compute the function $w(e; \source{I}, \sink{I})$ for all $e$ in linear time. 

\subsection{Tokens that move at most once}
In a reconfiguration sequence, there may be a token that moves at most once.
In other words, a token may not move at all or may move from the initial position to one of its out-neighbors and stay there.
Such a token causes an exception in our further discussion, and thus we want to remove such tokens in advance.
In this subsection, we show that such tokens are determined regardless of reconfiguration sequences, and that we can remove such tokens from the input without changing the reconfigurability.

First, we consider tokens that never move.
The following lemma states that such tokens are uniquely determined from the instance $(T, \source{I}, \sink{I})$ (not depending on an actual reconfiguration sequence).
From now on, we simply write $w(e)$ to denote $w(e; \source{I}, \sink{I})$.

\begin{lemma}\label{lem:rigid}
    Let $(T, \source{I}, \sink{I})$ be a yes-instance.
    For every reconfiguration sequence,
    the set of vertices containing tokens that do not move in the sequence is equal to
    \[
        R \coloneqq \inset{v \in \source{I} \cap \sink{I}}{w(e) = 0 \mbox{ for all } e \in \Gamma(v)}.
    \]
\end{lemma}
\begin{proof}
    Since for each $e \in A$, $w(e)$ equals the number of tokens that pass through $e$ by \cref{cor:length}, the token on $v \in R$ does not move.
    Moreover, if $v \notin R$, there exists an arc $e \in \Gamma(v)$ with $w(e) > 0$.
    Then, in every reconfiguration sequence, there is a token passing through $e$, implying that the token placed initially on $v$ must move to one of its out-neighbors.
\end{proof}

A token on a vertex $v \in R$ is said to be \emph{rigid}.
By~\cref{lem:rigid},
all rigid tokens do not move in any reconfiguration sequence.

For some $v \in R$, suppose that there exists $u \in N(v)$ such that $T$ has an arc $e \in \Gamma(u)$ with $w(e) > 0$.
If $(T, \source{I}, \sink{I})$ is a yes-instance, in any reconfiguration sequence, some token passes through $e$, implying that this token is placed on $u$ at some point.
However, since the token on $v$ is rigid, these tokens must be adjacent.
Thus, we obtain the following corollary.

\begin{corollary}\label{lem:rigid_exception}
    For $v \in R$, if there exists $u \in N(v)$ such that $T$ has an arc $e \in \Gamma(u)$ with $w(e) > 0$, then the instance $(T, \source{I}, \sink{I})$ is a no-instance.
\end{corollary}

Next, we consider tokens that move exactly once.
We can show that the set of arcs used by such tokens is uniquely determined from the instance $(T, \source{I}, \sink{I})$ (not depending on an actual reconfiguration sequence)
by a similar argument as in \cref{lem:rigid}.

\begin{lemma}\label{lem:blocking_arcs}
    Let $(T, \source{I}, \sink{I})$ be a yes-instance.
    For every reconfiguration sequence,
    the set of arcs used by tokens that move exactly once in the sequence
    is equal to
    \begin{align*}
        B \coloneqq \left\{e = (u, v) \in A \mid w(e) = 1 \text{ and } w(e') = 0
          \text{ for every } e' \in (\Gamma(u) \cup \Gamma(v)) \setminus \{e\}
         \right\}.
    \end{align*}
\end{lemma}
\begin{proof}
    Observe that, for every $e = (u, v) \in B$, we have $u \in \source{I} \setminus \sink{I}$ and $v \in \sink{I} \setminus \source{I}$.
    This follows from \cref{cor:length} and the fact that $w(e') = 0$ for every $e' \in (\Gamma(u) \cup \Gamma(v)) \setminus \{e\}$.
    Moreover, in any reconfiguration sequence, the token on $u$ slides to $v$ along $e$ and then stays there.
    Thus, the token moves exactly once along $e$.
    
    Conversely, suppose that $e = (u, v)$ is not in $B$.
    Then, either $w(e) = 0$ or there is an arc $e' \in (\Gamma(u) \cup \Gamma(v)) \setminus \{e\}$ with $w(e') > 0$.
    Let $B'$ be the set of arcs used by the tokens that move exactly once in at least one reconfiguration sequence.
    A goal here is to prove $e \notin B'$.
    If $w(e) = 0$, by \cref{cor:length}, every token does not pass through $e$, implying that $e \notin B'$.
    Moreover, if $u \notin \source{I}$, there are no tokens that slide exactly once along $e$, implying also that $e \notin B'$.
    Thus, we assume that there is an arc $e' \in (\Gamma(u) \cup \Gamma(v)) \setminus \{e\}$ with $w(e') > 0$ and $u \in \source{I}$.
    Suppose for contradiction that $e \in B'$.
    Then, in some reconfiguration sequence the token on $u$ slides to $v$ along $e$ and stay there.
    However, as $w(e') > 0$, there is a token passing through $e'$.
    Then these tokens must be adjacent, which is a contradiction.
\end{proof}

Take any $(u, v) \in B$.
By \cref{lem:blocking_arcs},
for every reconfiguration sequence
the token on $u$ slides to an out-neighbor $v$ of $u$ and stays there, which also implies $u \in \source{I}$ and $v \in \sink{I}$.
Since the tokens must form an independent set, the other tokens can be placed on neither $u$ nor $v$.
Given this, we refer to an arc $e \in B$ as a \emph{blocking arc} in $(T, \source{I}, \sink{I})$.

We can compute the rigid tokens and blocking arcs from $(T, \source{I}, \sink{I})$ in linear time.
Let $R$ be the set of vertices containing rigid tokens and $B$ the set of blocking arcs.
Let $T'$ be the polyforest obtained from $T$ by removing each arc $f$ such that $f$ is incident to a vertex in $R$, or
$f \notin B$ and
$f$ is incident to
an endpoint of $e \in B$.
Then, every component of $T'$ is either an isolated vertex in $R$, a component containing the two vertices connected by an arc in $B$, or a component without any rigid tokens or blocking arcs.
The following lemma reduces our problem to a slightly simplified one.

\begin{lemma}\label{lem:reduction}
    Suppose that for every $v \in R$, there is no $u \in N(v)$ such that $T$ has an arc $e \in \Gamma(u)$ with $w(e) > 0$.
    Then, $(T, \source{I}, \sink{I})$ is a yes-instance if and only if $(T', \source{I}, \sink{I})$ is a yes-instance.
\end{lemma}
\begin{proof}
    We first show the forward implication.
    For every $e \in A(T) \setminus A(T')$, we have $w(e) = 0$.
    This implies that no tokens pass through $e$, that is, every token moves inside a component of $T'$.
    Moreover, every independent set in $T$ is also an independent set of $T'$.
    Thus, every reconfiguration sequence for $(T, \source{I}, \sink{I})$ is also a reconfiguration sequence for $(T', \source{I}, \sink{I})$.
    
    Conversely, suppose that there is a reconfiguration sequence from $\source{I}$ to $\sink{I}$ in $T'$.
    Let $T_1, T_2, \ldots, T_c$ be the (weakly) connected components in $T'$.
    We may assume that this reconfiguration sequence is constructed in a component-wise manner: The tokens on $\source{I} \cap V(T_1)$ move to $\sink{I} \cap V(T_1)$, then the tokens on $\source{I} \cap V(T_2)$ move to $\sink{I} \cap V(T_2)$, and so on.
    For each $1 \le i \le c$, we let $S_i \coloneqq \sequence{I^i_0, I^i_1, \dots, I^i_{\ell_i}}$ be a reconfiguration sequence of $(T_i, \source{I} \cap V(T_i), \sink{I} \cap V(T_i))$.
    Now, we construct a reconfiguration sequence for $(T, \source{I}, \sink{I})$.
    For distinct components $T_i$ and $T_j$, we say that $T_i$ \emph{precedes} $T_j$ if
    \begin{enumerate}
        \item[(a)] $T_i$ is composed of a unique blocking arc $(x, y)$ and $T_j$ contains a vertex $z \in N_T(x)$, or
        \item[(b)] $T_j$ is composed of a unique blocking arc $(y, x)$ and $T_i$ contains a vertex $z \in N_T(x)$.
    \end{enumerate}
    Since $T$ is a polytree, this precedence relation is acyclic.
    By appropriately reindexing the components (and also sequences $S_i$), we can order the components $T_1, T_2, \ldots, T_c$ in such a way that there are no two components $T_i$ and $T_j$ such that $T_j$ precedes $T_i$ with $i < j$.
    Define a sequence of vertex sets $\sequence{I_0, I_1, \ldots, I_\ell}$ such that
    \begin{itemize}
        \item $I_0 = \source{I}$ and $I_\ell = \sink{I}$, and
        \item for $1 \le i \le c$ and for $1 \le j \le \ell_i$,
        $I_{p + j} = I_{p + j - 1} \setminus \{u\} \cup \{v\}$ with $u \in I^i_{j-1} \setminus I^i_{j}$ and $v \in I^i_{j} \setminus I^i_{j - 1}$, where $p = \sum_{1 \le k < i} \ell_k$. 
    \end{itemize}
    
    For each pair of two sets $(I_{p-1}, I_{p})$,
    we have $I_{p-1} \setminus I_{p} = \{u\}$ and $I_{p} \setminus I_{p - 1} = \{v\}$ for some $(u, v) \in A(T)$.
    Then, it suffices to show that each $I_p$ is an independent set in $T$ for $0 \le p \le \ell$.
    We prove this claim by induction on $p$.
    The base case $p = 0$ is trivial as $I_0 = \source{I}$.
    Let $p > 0$. 
    We assume that $I_{p - 1} \setminus I_{p} = \{u\}$ and $I_{p} \setminus I_{p - 1} = \{v\}$ for some $(u, v) \in A(T)$.
    By the construction of $I_p$, we have $\{u, v\} \subseteq V(T_j)$ for some $1 \le j \le c$.
    Since $I^j_k$ is an independent set in $T_j$ for all $k$, $I_p \cap V(T_j)$ is an independent set.
    Moreover, as $I_p \setminus \{v\} = I_{p-1} \setminus \{u\}$, $I_p \setminus \{v\}$ is an independent set.
    Thus, it suffices to show that $v$ has no neighbor in $I_{p - 1} \setminus \{u\}$.
    Suppose for contradiction that $v$ has a neighbor $x \in I_{p - 1} \setminus \{u\}$.
    As $I_p \cap V(T_j)$ is an independent set, $x$ belongs to some $V(T_{j'})$ for some $j' \neq j$.
    Since $w((u, v)) > 0$, by the assumption that every vertex in $R$ has no neighbor incident to an arc $e'$ with $w(e') > 0$, $x$ does not belong to $R$.
    We also have $v$ does not belong to $R$ as $w((u, v)) > 0$.
    Since the components are contained by removing all arcs incident to some vertex in $R$ and arcs incident to one of the end vertices of some arc in $B$, $(u, v)$ is a blocking arc or $x$ is an end vertex of an blocking arc.
    
    If $(u, v)$ is a blocking arc, by the precedence relation (b), we have $j' < j$.
    This inequality implies that $x \in \sink{I}$ from the construction of $I_p$. 
    By the definition of blocking arcs, $v \in \sink{I}$, which contradicts the fact that $\sink{I}$ is an independent set.
    Suppose otherwise that $x$ is an end vertex of an blocking arc.
    If $(x, y) \in B$, by (a), we have $j' < j$.
    By the definition of blocking arcs, we have $x \in \source{I}$.
    However, as $I_p$ has no vertices of $\source{I} \cap V(T_j)$, we have a contradiction that $x \in I_p$.
    Finally, if $(y, x) \in B$, by (b), we have $j < j'$.
    This directly implies that $x \notin I_p$, a contradiction.
    Hence, the lemma follows.
\end{proof}

It is easy to see that if $T'$ has more than one connected components, then we can solve the problem independently on each connected component.
Moreover, the problem is trivial on a component of size at most two.
Hence, we can assume that the input polytree has no vertices on which rigid tokens are placed or blocking arcs.

\subsection{Necessary and sufficient conditions for yes-instances}
In this subsection, we show the necessary and sufficient conditions for yes-instances.
By \cref{lem:reduction}, we may assume that the instance does not contain rigid tokens and blocking arcs.
For a yes-instance, pick some reconfiguration sequence and let $P_i$ be a directed path along which the $i$-th token moves in the reconfiguration sequence. $\mathcal{P} = \set{P_1, \dots, P_k}$ is obviously a directed path matching from $\source{I}$ to $\sink{I}$.
Since the instance does not contain rigid tokens and blocking arcs, we have $|P_i| \ge 2$ for every $i \in [k]$.
In addition, the successors of source vertices in $\mathcal P$ must be distinct.
To see this, observe that if the source vertices of two paths $P_i$ and $P_j$ in $\mathcal P$ have a common successor $x \coloneqq s'(P_i) = s'(P_j)$, then the tokens on $s(P_i)$ and $s(P_j)$ are both ``gazing'' the vertex $x$, and thus we cannot slide either of the tokens on $s(P_i)$ and $s(P_j)$ at all.
Therefore, $\mathcal P$ must satisfy that $s'(P_i) \neq s'(P_j)$ for all $i, j \in [k]$ with $i \neq j$.
Symmetrically, the predecessors of sink vertices must be distinct, that is, $t'(P_i) \neq t'(P_j)$ for all $i, j \in [k]$ with $i \neq j$.
The goal of this subsection is to show that these necessary conditions are also sufficient.

\begin{lemma}\label{lem:path_set}
    Let $(T, \source{I}, \sink{I})$ be an instance of \textsc{Directed Token Sliding} without rigid tokens and blocking arcs.
    Then $(T, \source{I}, \sink{I})$ is a yes-instance if and only if there exists a set $\mathcal{P} = \{P_1, \dots, P_k\}$ of directed paths satisfying all the following conditions:
    \begin{description}
        \item[{\rm (P1)}] $|P_i| \ge 2$ for all $i \in [k]$,
        \item[{\rm (P2)}] $\mathcal{P}$ is a directed path matching from $\source{I}$ to $\sink{I}$,
        \item[{\rm (P3)}] $s'(P_i) \neq s'(P_j)$ for all $i, j \in [k]$ with $i \neq j$, and
        \item[{\rm (P4)}] $t'(P_i) \neq t'(P_j)$ for all $i, j \in [k]$ with $i \neq j$. 
    \end{description}
\end{lemma}
We refer to the four conditions (P1)--(P4) as the \emph{path-set conditions}.

    We have already seen the only-if direction as above.
    In the following, we show the other direction.
    We call a vertex in a path other than the source or the sink an \emph{internal vertex}.
    For vertices $u, v \in V$, we define $\mathrm{dist}(u, v)$ as the length of the unique directed $u$-$v$ path in $T$ if such a path exists.
    We refer to a pair of directed paths $P$ and $P'$ satisfying the following conditions as a \emph{biased path pair}:
    \begin{itemize}
        \item $P$ and $P'$ have a common internal vertex $x$,
        \item $\mathrm{dist}(s(P), x) > \mathrm{dist}(s(P'), x)$ and $\mathrm{dist}(x, t(P)) > \mathrm{dist}(x, t(P'))$.
    \end{itemize}
    
    \begin{lemma}\label{clm:bias}
        If there exists a set of paths satisfying the path-set conditions, then there also exists a set of paths that satisfies the path-set conditions and does not include any biased path pair.
    \end{lemma}
    \begin{proof}
        Let $\mathcal{P} = \{P_1, \dots, P_k\}$ be a set of paths that satisfies the path-set conditions.
        For each $i \in [k]$, we let $s_i$ and $t_i$ be $s(P_i)$ and $t(P_i)$, respectively.
        Suppose that $\mathcal P$ contains a biased path pair $P_i$ and $P_j$,
        where $x$ is a common internal vertex of $P_i$ and $P_j$ that satisfies $\mathrm{dist}(s_i, x) > \mathrm{dist}(s_j, x)$ and $\mathrm{dist}(x, t_i) > \mathrm{dist}(x, t_j)$.
        We denote by $P_i[s_i, x]$ the subpath of $P_i$ from $s_i$ to $x$
        and by $P_i[x, t_i]$ that from $x$ to $t_i$.
        Similarly we define $P_j[s_j, x]$ and $P_j[x, t_j]$.
        Let $P_i'$ (resp.\ $P_j'$) denote the concatenation of $P_i[s_i, x]$ and $P_j[x, t_j]$ (resp.\ $P_j[s_j, x]$ and $P_i[x, t_i]$).
        Then we claim that $\mathcal{P}' \coloneqq \{P_1, \dots, P_k\} \setminus \{P_i, P_j\} \cup \{ P_i', P_j' \}$ also satisfies the path-set conditions.
        It is easy to verify that $\mathcal P'$ is a directed path matching from $\source{I}$ to $\sink{I}$; $\mathcal P'$ satisfies (P2).
        Since $x$ is an internal vertex,
        we also have $s'(P'_i) = s'(P_i)$, $t'(P'_i) = t'(P_j)$, $s'(P_j') = s'(P_j)$, and $t'(P'_j) = t'(P_i)$,
        and $|P_i[s_i, x]|, |P_i[x, t_i]|, |P_j[s_j, x]|, |P_j[x, t_j]|$ are all positive.
        Thus $\mathcal{P}'$ satisfies the conditions (P1), (P3), and (P4).
        Then
        \begin{align*}
            \sum_{P \in \mathcal{P}} |P|^2 - \sum_{P' \in \mathcal P'} |P'|^2 = |P_i|^2 + |P_j|^2 - |P'_i|^2 - |P'_j|^2 > 0,
        \end{align*}
        where the inequality is obtained from 
        the second condition of biased path pairs.
        The sum of squares of the length of paths in $\mathcal{P}'$ is strictly smaller than that of the original path set.
        Therefore, by applying this operation finitely many times, we can obtain a path set that satisfies the path-set conditions and does not contain any biased path pair. 
    \end{proof}
    
    In the following, let $\mathcal{P}^* = \{P_1^*, \dots, P_k^*\}$ be a set of paths that satisfies the path-set conditions and has no biased path pairs.
    We show that there exists a bijection $\pi: [k] \to [k]$ having the following property ($\ast$):
    \begin{description}
        \item[{\rm ($\ast$)}]
        $\source{I}$ is reconfigurable into $\sink{I}$ in a path-by-path manner following $\pi$;
        that is, by first moving the $\pi(1)$-th token from $s(P_{\pi(1)}^*)$ all the way to $t(P_{\pi(1)}^*)$ along $P_{\pi(1)}^*$, 
        then moving the $\pi(2)$-th token from $s(P_{\pi(2)}^*)$ all the way to $t(P_{\pi(2)}^*)$ along $P_{\pi(2)}^*$, and so on.
    \end{description}
    From now on, we consider how to construct $\pi$ satisfying $(\ast)$.
    For a vertex $v$ and a directed path $P$, we say that $v$ \emph{touches} $P$ or $P$ \emph{touches} $v$ if $N[v] \cap V(P) \neq \emptyset$, where $V(P)$ denotes the set of vertices in $P$.
    For $i \neq j$, let (A) and (B) be the following conditions:
    \begin{description}
        \item[{\rm (A)}] $s(P_i^*)$ touches $P_j^*$; 
        \item[{\rm (B)}] $t(P_j^*)$ touches $P_i^*$.
    \end{description}

    A binary relation $\dashleftarrow$ is defined as: $i \dashleftarrow j$ if and only if at least one of (A) and (B) holds for different $i$ and $j$.
    Let $G_{\dashleftarrow}$ be the directed graph such that
    the vertex set of $G_{\dashleftarrow}$ is $[k]$ and, for $i, j \in [k]$, $G_{\dashleftarrow}$ has the arc $(i, j)$ if and only if $j \dashleftarrow i$ holds.
    
    If $G_{\dashleftarrow}$ is a directed acyclic graph,
    then we can construct $\pi$ satisfying ($\ast$) as follows.
    Since $G_{\dashleftarrow}$ is a DAG, there is a vertex $i$ such that its out-degree is 0.
    For any $j \neq i$,
    every vertex in $P_i^*$ does not belong to $N[s(P_j^*)]$
    since $P_i^*$ does not touch $s(P_j^*)$,
    and every vertex in $P_j^*$ does not belong to $N[t(P_i^*)]$ since $P_j^*$ does not touch $t(P_i^*)$.
    The former implies that
    the token on $s(P_i^*)$ sliding to $t(P_i^*)$ along $P_i^*$ does not make adjacent token pairs and hence forms a reconfiguration sequence from $\source{I}$ to $\source{I} \setminus \{ s(P_i^*)\} \cup \{t(P_i^*)\}$.
    The latter implies that
    the graph $G_{\dashleftarrow}$ for the resulting independent set $\source{I} \setminus \{ s(P_i^*)\} \cup \{t(P_i^*)\}$
    is obtained by deleting the vertex $i$,
    which is still a DAG.
    By repeating the above procedure,
    $\source{I}$ is reconfigurable into $\sink{I}$ in a path-by-path manner.
    Thus a bijection $\pi : [k] \to [k]$ satisfying $\pi(1) < \pi(2) < \cdots < \pi(k)$
    with respect to
    a topological order of $G_{\dashleftarrow}$
    admits ($\ast$), as required.

    The following lemma says that $G_{\dashleftarrow}$ is actually a directed acyclic graph,
    which verifies the if-condition of \cref{lem:path_set}.
    \begin{lemma}\label{lem:dag}
    $G_{\dashleftarrow}$ is a directed acyclic graph.
    \end{lemma}
    \begin{proof}
    Suppose, to the contrary, that there is a directed cycle in $G_{\dashleftarrow}$.
    We can assume that $(1,2, \dots, \ell, 1)$ is a minimal one.
    For convenience, the addition $+$ and subtraction $-$ are taken over modulo $\ell$,
    i.e., $\ell + 1$ is regarded as $1$ and $0$ is regarded as $\ell$.
    
    Suppose moreover that $P_i^*$ and $P_{i+1}^*$ have a common internal vertex for all $i \in [k]$.
    Let $x$ and $y$ be the source and the sink of the maximal common subpath of $P^*_i$ and $P^*_{i + 1}$, respectively.
    Since $T$ is a polytree, these $x$ and $y$ are uniquely determined. 
    For $i \in [\ell]$, at least one of (A) and (B) holds.
    If (A) $s(P^*_i)$ touches $P^*_{i+1}$, then we have $\mathrm{dist}(s(P_i^*), x) \le 1$.
    If $\mathrm{dist}(s(P_i^*), x) \ge \mathrm{dist}(s(P_{i+1}^*), x)$, either $s'(P^*_i) = s'(P^*_{i+1})$ or $s(P^*_i)$ and $s(P^*_{i + 1})$ are adjacent, which do not hold as
    (P3)
    or the fact that $\source{I}$ is an independent set.
    Thus, we have $\mathrm{dist}(s(P_i^*), x) < \mathrm{dist}(s(P_{i+1}^*), x)$.
    If (B) $t(P^*_{i + 1})$ touches $P^*_{i}$, then we have $\mathrm{dist}(y, t(P_{i+1}^*)) \le 1$.
    Then, by a symmetric argument, 
    we have $\mathrm{dist}(y, t(P_i^*)) > \mathrm{dist}(y, t(P_{i+1}^*))$.
    Let $z_i$ be an arbitrary common internal vertex of $P_i^*$ and $P_{i+1}^*$.
    The above argument is summarised as
    \begin{align}\label{eq:<}
        \mathrm{dist}(s(P_i^*), z_i) < \mathrm{dist}(s(P_{i+1}^*), z_i)
        \mbox{ or }
        \mathrm{dist}(z_i, t(P_i^*)) > \mathrm{dist}(z_i, t(P_{i+1}^*)).
    \end{align}
    \Cref{eq:<} and the non-existence of biased pairs together imply that for every $i$, 
    we have
    \begin{align}\label{eq:<=}
        \mathrm{dist}(s(P_i^*), z_i) \le \mathrm{dist}(s(P_{i+1}^*), z_i)
        \mbox{ and }
        \mathrm{dist}(z_i, t(P_i^*)) \ge \mathrm{dist}(z_i, t(P_{i+1}^*)).
    \end{align}
    For a walk $W$ in $\und{T}$, let $\overrightarrow{W}$ be an oriented walk obtained from $W$ by orienting each edge from (arbitrary) one of the end vertices to the other.
    Let $\mathrm{fwd}(\overrightarrow{W})$ and $\mathrm{rev}(\overrightarrow{W})$ be the numbers of times that $\overrightarrow{W}$ passes through arcs in the forward and reverse directions in $T$, respectively.
    Here, by the fact that $\mathrm{dist}(s(P_i^*), z_i) \le \mathrm{dist}(s(P_{i+1}^*), z_i)$ (\cref{eq:<=}),
    there exists an oriented walk $\overrightarrow{W}$ from $s(P_i^*)$ to $s(P_{i+1}^*)$ satisfying $\mathrm{fwd}(\overrightarrow{W}) \le \mathrm{rev}(\overrightarrow{W})$.
    This can be obtained by traversing $\und{T}$ from $s(P_i^*)$ to $s(P_{i+1}^*)$ via $z_i$ (which we denote by $s(P^*_i) \to z_i \to s(P^*_{i+1})$).
    In particular, if $\mathrm{dist}(s(P_i^*), z_i) < \mathrm{dist}(s(P_{i+1}^*), z_i)$, then $\mathrm{fwd}(\overrightarrow{W}) < \mathrm{rev}(\overrightarrow{W})$ holds. 
    Suppose that $\mathrm{dist}(s(P_1^*), z_1) < \mathrm{dist}(s(P_{2}^*), z_1)$ holds.
    Then, the closed oriented walk 
    \[
        \overrightarrow{W_s} \coloneqq s(P_1^*) \to z_1 \to s(P_2^*) \to \dots \to s(P_\ell^*) \to z_\ell \to s(P_1^*)
    \]
    satisfies $\mathrm{fwd}(\overrightarrow{W_s}) < \mathrm{rev}(\overrightarrow{W_s})$. See \Cref{fig:img5} for an illustration.
    This implies that there is an arc $e$ in $T$ such that $e$ occurs in $\overrightarrow{W_s}$ at least once and
    the number of occurrences of $e$ is strictly smaller than that of the reverse of $e$ in $\overrightarrow{W_s}$.
    This contradicts the fact that $T$ is a polytree.
    Thus suppose $\mathrm{dist}(s(P_1^*), z_1) \ge \mathrm{dist}(s(P_{2}^*), z_1)$ holds.
    Then $\mathrm{dist}(z_1, t(P_1^*)) > \mathrm{dist}(z_1, t(P_{2}^*))$ follows from~\cref{eq:<}.
    By a similar argument as above, we can construct a closed oriented walk $\overrightarrow{W_t}$ from $t(P_1^*)$ to $t(P_1^*)$ satisfying $\mathrm{fwd}(\overrightarrow{W_t}) > \mathrm{rev}(\overrightarrow{W_t})$.
    This also contradicts the fact that $T$ is a polytree.
    Thus, we derive a contradiction, assuming that $G_{\dashleftarrow}$ has a directed cycle and $P_i^*$ and $P_{i+1}^*$ have a common internal vertex for all $i \in [k]$.
    
    \begin{figure}[t]
        \begin{minipage}{0.6\linewidth}
            \centering
            \newcommand{\basePath}[1]{
    \node[draw=none, fill=none] at (0, 6) (pLabel) {\huge $P^*_{#1}$};
    \node at (0, 5.2) (topNode) {};
    \node at (0, -.5) (bottomNode) {};
}

\def\exscale{0.5}
\begin{tikzpicture}[scale = \exscale,
    every node/.style = {scale=\exscale, circle, line width=0.5pt, draw=black, fill=white, minimum size=12pt, inner sep=0pt, outer sep=0pt},
    walkHead/.style = {arrows = {-stealth[length=3pt]}},
    walk/.style = {walkHead, draw=red, line width = 0.5pt},
    arc/.style = {walkHead, draw=black, line width = 0.5pt},
    ]
    \def\myXshift{4}
    \def\walkShift{10pt, 0}
    \def\walkLenShift{19pt, 0}
    \begin{scope}
        \basePath{1}
        \node[label={[label distance = 3pt] 225:\huge $z_1$}] at (0, 3) (x11) {};
        \draw[arc] (topNode) to (x11);
        \draw[arc] (x11) to (bottomNode);

        \draw [decorate,decoration={brace,amplitude=3pt}]
        ($(topNode.south) + (\walkLenShift)$) -- ($(x11.north) + (\walkLenShift)$);

        \node[draw=none, fill=none] at ($(topNode.south)!0.5!(x11.north) + (\myXshift/2, 0)$) {\huge $\le$};

        \draw[walk] ($(topNode.south) + (\walkShift)$) to ($(x11.north) + (\walkShift)$);
        \draw[walk] ($(x11.east) + (3pt, 0)$) to ($(x11.west) + (\myXshift, 0) - (3pt, 0)$);
    \end{scope}

    \begin{scope}[shift={(\myXshift, 0)}]
        \basePath{2}
        \node[label={[label distance = 3pt] 225:\huge $z_1$}] at (0, 3) (x12) {};
        \node[label={[label distance = 3pt] 225:\huge $z_2$}] at (0, 1.5) (x22) {};
        \draw[arc] (topNode) to (x12);
        \draw[arc] (x12) to (x22);
        \draw[arc] (x22) to (bottomNode);

        \draw [decorate,decoration={brace,amplitude=3pt}]
        ($(x12.north) - (\walkLenShift)$) -- ($(topNode.south) - (\walkLenShift)$);
        \draw [decorate,decoration={brace,amplitude=3pt}]
        ($(topNode.south) + (\walkLenShift)$) -- ($(x22.north) + (\walkLenShift)$);

        \node[draw=none, fill=none] at ($(topNode.south)!0.5!(x22.north) + (\myXshift/2, 0)$) {\huge $\le$};

        \draw[walk] ($(x12.north) - (\walkShift)$) to ($(topNode.south) - (\walkShift)$);
        \draw[walk] ($(topNode.south) + (\walkShift)$) to ($(x22.north) + (\walkShift)$);
        \draw[walk] ($(x22.east) + (3pt, 0)$) to ($(x22.west) + (\myXshift, 0) - (3pt, 0)$);
    \end{scope}

    \begin{scope}[shift={(2*\myXshift, 0)}]
        \basePath{3}
        \node[label={[label distance = 3pt] 225:\huge $z_2$}] at (0, 1.5) (x13) {};
        \node[draw=none, fill=none] at (0, 3.5) (xn3) {};
        \draw[arc] (topNode) to (x13);
        \draw[arc] (x13) to (bottomNode);
        \draw [decorate,decoration={brace,amplitude=3pt}]
        ($(x13.north) - (\walkLenShift)$) -- ($(topNode.south) - (\walkLenShift)$);

        \draw[walk] ($(x13.north) - (\walkShift)$) to ($(topNode.south) - (\walkShift)$);
        \draw[walk] ($(topNode.south) + (\walkShift)$) to ($(xn3.north) + (\walkShift)$);
        \draw[draw=red, dashed, line width=0.5pt] ($(xn3.north) + (\walkShift) - (0, 0.2)$) to ($(xn3.north) + (\walkShift) - (0, 1.5)$);
    \end{scope}

    \begin{scope}[shift={(3*\myXshift, 0)}]
        \basePath{\ell}
        \node[label={[label distance = 3pt] 225:\huge $z_{\ell-1}$}] at (0, 3.5) (xl14) {};
        \node[label={[label distance = 3pt] 225:\huge $z_{\ell}$}] at (0, 2) (xl4) {};
        \draw[arc] (topNode) to (xl14);
        \draw[arc] (xl14) to (xl4);
        \draw[arc] (xl4) to (bottomNode);

        \draw [decorate,decoration={brace,amplitude=3pt}]
        ($(topNode.south) + (\walkLenShift)$) -- ($(xl4.north) + (\walkLenShift)$);

        \node[draw=none, fill=none] at ($(topNode.south)!0.5!(xl4.north) + (\myXshift/2, 0)$) {\huge $\le$};

        \draw[draw=black, dotted, line width=2pt] ($(xl14.north) - (\walkShift) - (2, 2)$) to ($(xl14.north) - (\walkShift) - (1, 2)$);

        \draw[walkHead, draw=red, dashed, line width=0.5pt] ($(xl14.north) - (\walkShift) - (1, 0)$) to ($(xl14.north) - (\walkShift) - (3pt, 0)$);
        \draw[walk] ($(xl14.north) - (\walkShift)$) to ($(topNode.south) - (\walkShift)$);
        \draw[walk] ($(topNode.south) + (\walkShift)$) to ($(xl4.north) + (\walkShift)$);

        \draw[walk] ($(xl4.east) + (3pt, 0)$) to ($(xl4.west) + (\myXshift, 0) - (3pt, 0)$);
    \end{scope}

    \begin{scope}[shift={(4*\myXshift, 0)}]
        \basePath{1}
        \node[label={[label distance = 3pt] 225:\huge $z_{\ell}$}] at (0, 2) (xl5) {};
        \draw[arc] (topNode) to (xl5);
        \draw[arc] (xl5) to (bottomNode);

        \draw [decorate,decoration={brace,amplitude=3pt}]
        ($(xl5.north) - (\walkLenShift)$) -- ($(topNode.south) - (\walkLenShift)$);

        \draw[walk] ($(xl5.north) - (\walkShift)$) to ($(topNode.south) - (\walkShift)$);
    \end{scope}
\end{tikzpicture}
            \caption{Closed oriented walk $s(P_1^*) \to z_1 \to s(P_2^*) \to \dots \to s(P_\ell^*) \to z_\ell \to s(P_1^*)$.}
            \label{fig:img5}
        \end{minipage}
        \hfill
        \begin{minipage}{0.35\linewidth}
            \centering
            \def\exscale{0.5}
\begin{tikzpicture}[scale = \exscale,
                every node/.style = {scale=\exscale, circle, line width=0.5pt, draw=black, fill=black, minimum size=8pt, inner sep=0pt, outer sep=0pt},
        ]
        \def\myXshift{3.4}
        \def\myYshift{1.3}

        \node[label ={[label distance=10pt]1:\huge$s'(P^*_i)$}] at (1, 1) (sppi) {};
        \node[label ={[label distance=10pt]1:\huge$s(P^*_i)$}]  at ([shift=({90 :2 cm})]sppi) (c01) {};
        \node[label = {[label distance=4pt] 90: \huge$x^*$}]  at ([shift=({210:2 cm})]sppi) (c11)  {};
        \node  at ([shift=({250:2 cm})]sppi) (c21)  {};
        \node  at ([shift=({290:2 cm})]sppi) (c31)  {};
        \node  at ([shift=({330:2 cm})]sppi) (c41)  {};

        \foreach \i/\rot in {0/60, 1/180,2/220,3/260,4/300} {
                        \node[draw=none, minimum size=0pt] at ([shift=({ 0+\rot:2 cm})]c\i1) (c\i2)  {};
                        \node[draw=none, minimum size=0pt] at ([shift=({60+\rot:2 cm})]c\i1) (c\i3)  {};
                        \draw (c\i1) -- (c\i2)  -- (c\i3) -- (c\i1) -- cycle;
                        \draw (sppi) -- (c\i1);
                }
        \node[draw=none, fill=none] at ($(c11) + (0.2, 1)$) (c14) {};
        \node[draw=none, fill=none] at ($(c14) + (-2, 0.2)$) (c15) {\huge$C_{x^*}$};
        \draw[thick, dashed] \convexpath{c11,c13,c12,c14}{1em};
\end{tikzpicture}
            \caption{An illustration of a polytree considered in the proof of \cref{lem:dag}.}
            \label{fig:img7}
        \end{minipage}
    \end{figure}
    
    The remaining task is to show that $P_i^*$ and $P_{i+1}^*$ have a common internal vertex for all $i$ under the assumption that $G_{\dashleftarrow}$ has a directed cycle.
    Suppose to the contrary that $P_i^*$ and $P_{i+1}^*$ have no common internal vertex.
    We only consider the case where the condition (A) in the definition of $\dashleftarrow$ holds for $i \dashleftarrow i+1$, that is, $P^*_{i+1}$ touches $s(P^*_i)$;
    the case for (B) is symmetric.
    As $|P_i^*| \geq 2$, $P_i^*$ has an arc $(s'(P_i^*), x^*)$.
    See \Cref{fig:img7} for an illustration.
    For $x \in N_T(s'(P_i^*))$,
    we denote by $C_x$ the weakly connected component containing $x$ in the polyforest obtained from $T$ by deleting the arc $(s'(P_i^*), x)$.
    Then observe that $V(P^*_{i+1}) \cap C_{x^*} = \emptyset$.
    To see this, if $V(P^*_{i+1}) \cap C_{x^*} \neq \emptyset$, then $P_{i+1}^*$ must have the arc $(s'(P_i^*), x^*)$ as $P^*_{i+1}$ touches $s(P^*_i)$.
    This particularly implies that $s'(P_i^*)$ belongs to $P_{i+1}^*$ and is different from $t(P_{i+1}^*)$ due to the direction of arc $(s'(P^*_i), x^*)$.
    Since $\source{I}$ is an independent set, we have $s'(P_i^*) \neq s(P_{i+1}^*)$.
    Thus $P_i^*$ and $P_{i+1}^*$ have a common internal vertex $s'(P_i^*)$,
    contradicting the assumption that
    $P_i^*$ and $P_{i+1}^*$ have no common internal vertex.
    Hence we obtain $V(P^*_{i+1}) \cap C_{x^*} = \emptyset$.
    
    Suppose $i - 1 = i + 1$, i.e., $i +1 \dashleftarrow i$ (which implies $\ell = 2$).
    Then $P_i^*$ touches $s(P_{i+1}^*)$ or $P_{i+1}^*$ touches $t(P_i^*)$.
    In the former case,
    since $V(P^*_{i+1}) \cap C_{x^*} = \emptyset$ and $\source{I}$ is an independent set,
    $s(P_{i+1}^*)$ must belong to $N_T(s'(P_i^*))$.
    Then we have $s'(P_{i+1}^*) = s'(P_{i}^*)$ by $i \dashleftarrow i+1$,
    which contradicts (P3).
    In the latter case, since $V(P_{i+1}^*) \cap C_{x^*} = \emptyset$,
    we have $x^* = t(P_i^*)$ and $s'(P_i^*) \in V(P_{i+1}^*)$.
    Moreover, neither $s'(P_i^*) = s(P_{i+1}^*)$ nor $s'(P_i^*) = t(P_{i+1}^*)$,
    since $\source{I}$ and $\sink{I}$ are independent sets.
    Thus $s'(P_i^*)$ must be an internal vertex of $P_{i+1}^*$, contradicting the assumption.
    
    Suppose $i - 1 \neq i + 1$ (which implies $\ell \geq 3$).
    Observe that $V(P^*_{i-1}) \cap C_{s(P_i^*)} = \emptyset$.
    To see this, suppose $V(P^*_{i-1}) \cap C_{s(P_i^*)} \neq \emptyset$.
    As $i - 1 \dashleftarrow i$, either $s(P^*_{i-1})$ touches $P^*_{i}$ or $t(P^*_{i})$ touches $P^*_{i-1}$.
    In both cases, $s(P^*_{i})$ touches $P^*_{i - 1}$ and hence $i \dashleftarrow i - 1$, contradicting the minimality of the cycle $(1,2,\dots,\ell,1)$.
    Observe also that $s'(P_i^*) \notin V(P^*_{i-1})$
    as otherwise we obtain $i \dashleftarrow i - 1$,
    which again contradicts the minimality of the cycle $(1,2,\dots,\ell,1)$.
    Here we additionally distinguish the two cases:
    (i) $s'(P_{i}^*) \in V(P_{i+1}^*)$ and (ii) $s'(P_{i}^*) \notin V(P_{i+1}^*)$.
    
    (i) $s'(P_{i}^*) \in V(P_{i+1}^*)$.
    Since $P_i^*$ and $P_{i+1}^*$ have no common internal vertices,
    we have $t(P_{i+1}^*) = s'(P_i^*)$.
    By the minimality of the cycle $(1,2,\dots,\ell,1)$,
    $P_{i-1}^*$ has none of vertices in $N_T[s'(P_i^*)]$.
    This implies, together with $i - 1 \dashleftarrow i$ and $V(P^*_{i-1}) \cap C_{s(P_i^*)} = \emptyset$, that $V(P_{i-1}^*) \subseteq C_{x^*}$
    and $x^* \notin V(P_{i-1}^*)$.
    By $i+1 \dashleftarrow i+2 \dashleftarrow \cdots \dashleftarrow i-1$,
    there is an index $m$ with $i-1 \neq m \neq i+1$
    such that $x^* \in V(P^*_m)$.
    This implies $m \dashleftarrow i+1$,
    contradicting the minimality of the cycle.
    
    (ii) $s'(P_{i}^*) \notin V(P_{i+1}^*)$.
    In this case, $P_{i+1}^*$ has a vertex in $N[s(P_i^*)] \setminus \{ s'(P_i^*) \}$ and does not have the arc $(s(P_i^*), s'(P_i^*))$.
    Thus we have $V(P_{i+1}^*) \subseteq C_{s(P_i^*)}$.
    By $i+1 \dashleftarrow i+2 \dashleftarrow \cdots \dashleftarrow i-1$,
    there is an index $m$ with $i-1 \neq m \neq i+1$
    such that $s'(P_i^*) \in V(P^*_m)$.
    This implies $i \dashleftarrow m$,
    contradicting the minimality of the cycle.
    
    This completes the proof.
    \end{proof}

\subsection{Algorithms}\label{sec:algo}
In this subsection, we provide an algorithm for checking the reconfigurability in $\bigO{|V|}$ time,
and that for constructing a reconfiguration sequence in $\bigO{|V|^2}$ time (if the answer is affirmative), proving \cref{thm:polytree}.

For $U \subseteq V$,
we define $N^+ (U)$ (resp.\ $N^-(U)$) as $N^+ (U) \coloneqq \bigcup_{u \in U} N^+(u) \setminus U$ (resp.\ $N^- (U) \coloneqq \bigcup_{u \in U} N^-(u) \setminus U$).
For a mapping $f \colon \source{I} \to N^+(\source{I})$,
let $f(\source{I})$ denote the image of $f$, i.e.,
$f(\source{I}) \coloneqq \inset{ f(x) }{x \in \source{I}}$.
Similarly, the image of $g \colon \sink{I} \to N^-(\sink{I})$ is denoted as $g(\sink{I})$.

\subsubsection{Algorithm for checking the reconfigurability}
\label{sec:check_fast}

Suppose that there is a set $\mathcal{P} = \set{P_1, \dots, P_k}$ of directed paths satisfying the path-set conditions.
Then the set of paths obtained from $\mathcal{P}$ by exchanging each $P_i$ with the subpath from $s'(P_i)$ to $t'(P_i)$ is a directed path matching from $\inset{s'(P_i)}{i \in [k]}$ to $\inset{t'(P_i)}{i \in [k]}$.
In this case, the mappings $f \colon \source{I} \to N^+(\source{I})$ and $g \colon \sink{I} \to N^-(\sink{I})$ defined by $f(s(P_i)) \coloneqq s'(P_i)$ and $g(t(P_i)) \coloneqq t'(P_i)$, respectively,
satisfy the following four conditions:
\begin{description}
    \item[{\rm (C1)}] $f$ and $g$ are injective,
    \item[{\rm (C2)}] $(s, f(s)) \in A$ for all $s \in \source{I}$,
    \item[{\rm (C3)}] $(g(t), t) \in A$ for all $t \in \sink{I}$, and
    \item[{\rm (C4)}] $w(e; f(\source{I}), g(\sink{I})) \geq 0$ for each arc $e$.
\end{description}
In particular,
(C1) follows from (P3) and (P4),
and (C4) follows from \cref{lem:chara_matching}.

Conversely, if there are mappings $f \colon \source{I} \to N^+(\source{I})$ and $g \colon \sink{I} \to N^-(\sink{I})$ satisfying the above four conditions (C1)--(C4),
then we can construct a set of directed paths satisfying the path-set conditions (P1)--(P4) as follows.
By \cref{lem:chara_matching} and (C4), there exists a directed path matching $\mathcal{P}' = \set{P'_1, \dots, P'_k}$ from $f (\source{I} )$ to $g(\sink{I} )$.
Define a set $\mathcal{P} = \set{P_1, \dots, P_k}$ of paths from $\mathcal{P}'$ by appending $f^{-1}(s(P'_i) )$ and $g^{-1}(t(P'_i) )$ before the source and after the sink for each $P'_i$, respectively.
Then $|P_i| \ge 2$ for each $i \in [k]$, and $\mathcal P$ is a directed path matching from $\source{I}$ and $\sink{I}$.
Moreover,
since $\mathcal P'$ is a directed path matching from $f(\source{I})$ to $g(\sink{I})$, we have $s'(P_i) \neq s'(P_j)$ and $t'(P_i) \neq t'(P_j)$ for $i \neq j$, implying that $\mathcal P$ satisfies the path-set conditions.

By the above argument, we obtain the following necessary and sufficient conditions for the reconfigurability:
\begin{lemma}\label{lem:ns_path_set}
    For an instance without rigid tokens and blocking arcs,
    there exists a set of paths satisfying the path-set conditions if and only if there exist mappings $f \colon \source{I} \to N^+(\source{I})$ and $g \colon \sink{I} \to N^-(\sink{I})$ satisfying {\rm (C1)}--{\rm (C4)}.
\end{lemma}

By \cref{lem:path_set,lem:ns_path_set},
for checking the reconfigurability
it suffices to compute $f$ and $g$ satisfying (C1)--(C4) or determine that such $f$ and $g$ do not exist;
it can be done in $\bigO{|V|}$ time in a greedy manner as follows. 
In the following description, we regard $T$ as a rooted tree with an arbitrary root. 
We initialize the values of $f$ and $g$ as $f(s) \coloneqq \bot$ for all $s \in \source{I}$ and $g(t) \coloneqq \bot$ for all $t \in \sink{I}$ ($\bot$ means ``undefined'').
To let $f$ and $g$ satisfy the conditions in \Cref{lem:ns_path_set}, we define each value of $f$ and $g$ one by one.
We write $\source{I}_{\mathrm{undef}} \subseteq \source{I}$ (resp.\ $\source{I}_{\mathrm{def}} \subseteq \source{I}$) as the set of the vertices for which the values of $f$ have been undefined (resp.\ defined).
We also define $\sink{I}_{\mathrm{undef}}$ and $\sink{I}_{\mathrm{def}}$ in the same way.
Throughout the following process to determine $f$ and $g$, we keep an invariant that the conditions (C1)--(C3)
hold for $s \in \source{I}_{\mathrm{def}}$ and $t \in \sink{I}_{\mathrm{def}}$, moreover, for each arc $e$, $w(e; \source{I}_{\mathrm{undef}} \cup f(\source{I}_{\mathrm{def}}), \sink{I}_{\mathrm{undef}} \cup g(\sink{I}_{\mathrm{def}})) \geq 0$.

In the following, we describe each step of the algorithm.
Let $v^* \in  \source{I}_{\mathrm{undef}} \cup \sink{I}_{\mathrm{undef}}$ be a vertex with the largest depth among $\source{I}_{\mathrm{undef}} \cup \sink{I}_{\mathrm{undef}}$ in the rooted tree.
If $v^* \in \source{I}_{\mathrm{undef}}$,
we define
\begin{align*}
    N_{\mathrm{cand}}(v^*) \coloneqq \inset{u \in N^+(v^*) \setminus f(\source{I}_{\mathrm{def}})}{w((v^*, u); \source{I}_{\mathrm{undef}} \cup f(\source{I}_{\mathrm{def}}), \sink{I}_{\mathrm{undef}} \cup g(\sink{I}_{\mathrm{def}})) \geq 1}.
\end{align*}
Otherwise, i.e., $v^* \in \sink{I}_{\mathrm{undef}} \setminus \source{I}_{\mathrm{undef}}$,
define
\begin{align*}
    N_{\mathrm{cand}}(v^*) \coloneqq \inset{u \in N^-(v^*) \setminus g(\sink{I}_{\mathrm{def}})}{w((u, v^*); \source{I}_{\mathrm{undef}} \cup f(\source{I}_{\mathrm{def}}), \sink{I}_{\mathrm{undef}} \cup g(\sink{I}_{\mathrm{def}})) \geq 1}.
\end{align*}
If $N_{\mathrm{cand}}(v^*)$ is empty,
terminate the execution.
Otherwise, choose $u^* \in N_{\mathrm{cand}}(v^*)$ with the largest depth among $N_{\mathrm{cand}}(v^*)$ and define $f(v^*) \coloneqq u^*$ if $v^* \in \source{I}_{\mathrm{undef}}$ and define $g(v^*) \coloneqq u^*$ otherwise.
When $f(v^*)$ or $g(v^*)$ is newly defined in this step, we accordingly update $\source{I}_{\mathrm{def}}$, $\source{I}_{\mathrm{undef}}$, $\sink{I}_{\mathrm{def}}$, and $\sink{I}_{\mathrm{undef}}$.

Note that each vertex in $\source{I} \cup \sink{I}$ is selected as $v^*$ at most twice.
After the execution of the above algorithm, output $f$ and $g$ if $f(s) \neq \bot$ for all $s \in \source{I}$ and $g(t) \neq \bot$ for all $t \in \sink{I}$; otherwise, we conclude that no such $f$ and $g$ exist.
In each iteration, only one arc $(v^*, u^*)$ or $(u^*, v^*)$ decreases its $w$-value by $1$: If $v^* \in \source{I}_{\mathrm{undef}}$, 
\begin{align*}
    &w(e; \source{I}_{\mathrm{undef}} \setminus \{v^*\} \cup f(\source{I}_{\mathrm{def}}) \cup \{u^*\}, \sink{I}_{\mathrm{undef}} \cup g(\sink{I}_{\mathrm{def}})) \\
    &=
    \begin{cases}
       w(e; \source{I}_{\mathrm{undef}} \cup f(\source{I}_{\mathrm{def}}), \sink{I}_{\mathrm{undef}} \cup g(\sink{I}_{\mathrm{def}})) & \text{if $e\neq (v^*, u^*)$}\\
       w(e; \source{I}_{\mathrm{undef}} \cup f(\source{I}_{\mathrm{def}}), \sink{I}_{\mathrm{undef}} \cup g(\sink{I}_{\mathrm{def}})) - 1 & \text{if $e =  (v^*, u^*)$};
    \end{cases}
\end{align*}
and otherwise $w(e; \source{I}_{\mathrm{undef}}  \cup f(\source{I}_{\mathrm{def}}), \sink{I}_{\mathrm{undef}} \setminus \{v^*\} \cup g(\sink{I}_{\mathrm{def}}) \cup \{u^*\})$ can be computed in a symmetric fashion.
Using this property, we can implement the algorithm that runs in $\bigO{|V|}$ time.
More precisely, we maintain the values of $w$ and the indicator functions of $\source{I}_{\mathrm{def}}$ and $\sink{I}_{\mathrm{def}}$ in tables.
By referring the tables, we can compute $N_{\mathrm{cand}}(v^*)$ in $\bigO{|N(v^*)|}$ time.
Moreover, we can update the tables in $\bigO{1}$ time.
Therefore each iteration runs in $\bigO{|N(v^*)|}$ time, and thus the algorithm runs in total in $\bigO{|V|}$ time
since each edge is touched twice.

\begin{lemma}
    If there exist mappings $f$ and $g$ that satisfy the conditions in \Cref{lem:ns_path_set}, then the algorithm described above correctly outputs one of such mappings. Otherwise the algorithm reports that no such mappings exist.
\end{lemma}

\begin{proof}
    By the invariant that we keep, it is clear that the algorithm outputs mappings $f$ and $g$ satisfying the conditions in \Cref{lem:ns_path_set}.
    We claim that our additional definition of $f$ or $g$ in each step of the algorithm does not compromise the possibility that $f$ and $g$ can be extended to desired mappings by further definitions. 
    In each step of the algorithm, assume that $N_{\mathrm{cand}}(v^*)$ is not empty and that there exist desired mappings $f'$ and $g'$ with $f'(s) = f(s)$ for all $s \in \source{I}_{\mathrm{def}}$ and $g'(t) = g(t)$ for all $t \in \sink{I}_{\mathrm{def}}$. 
    We also assume $v^* \in \source{I}_{\mathrm{undef}}$ (as the case $v^* \in \sink{I}_{\mathrm{undef}} \setminus \source{I}_{\mathrm{undef}}$ is symmetric).
    Then we show that there also exist desired $f''$ and $g''$ with $f''(s) = f(s)$ for all $s \in \source{I}_{\mathrm{def}}$, $g''(t) = g(t)$ for all $t \in \sink{I}_{\mathrm{def}}$, and $f''(v^*) = u^*$.
    If $f'(v^*) = u^*$, we are done.
    We consider the other case, $f'(v^*) = x^*$ with $x^* \neq u^*$. 
    Observe that in this case, the depth of $u^*$ is exactly $1$ greater than that of $v^*$.
    This is because of the facts $x^* \in N_{\mathrm{cand}}(v^*)$ and 
    that $N_{\mathrm{cand}}(v^*)$ includes at most one vertex with the depth exactly $1$ smaller than that of $v^*$ by the structure of the rooted polytree $T$.
    Since we define each value of $f$ and $g$ in descending order of depths of vertices, the only way to decrease the $w$ value of $(v^*, u^*)$ when extending the definitions of current $f$ and $g$ is to set $f(v^*) = u^*$. Thus $w((v^*, u^*); f'(\source{I}), g'(\sink{I})) > 0$.
    Similarly, since values of $f$ for vertices in $N(u^*) \cap \source{I}$ except for $v^*$ have been already defined, we have $u^* \notin f'(\source{I})$.
    By these facts, we can obtain mappings $f''$ and $g''$ satisfying the conditions in \Cref{lem:ns_path_set} by setting $g'' \coloneqq g'$, $f''(v^*) \coloneqq u^*$, and $f''(v) \coloneqq f'(v)$ for all $v \in \source{I} \setminus \{v^*\}$.    
    Because of the above arguments, if the input is a yes-instance, the algorithm does not terminate in the middle and at last outputs desired $f$ and $g$. (If the input is a no-instance, obviously the algorithm terminates before the whole definition.)
\end{proof}

\subsubsection{Algorithm for constructing a reconfiguration sequence}\label{sec:construct_fast}
By \cref{cor:length}, if the instance is a yes-instance, all reconfiguration sequences have the same length.
In this subsection, we present an algorithm to construct one of them in $\bigO{|V|^2}$ time.
We first assume that two mappings $f$ and $g$ satisfying conditions (C1)--(C4) have already been computed.
If such $f$ and $g$ do not exist, the instance is a no-instance by \Cref{lem:ns_path_set}.
Here, our target is to fill in gaps between $f ( \source{I} )$ and $g ( \sink{I} )$ avoiding biased path pairs.
For preparation of further arguments, we define \emph{weakly biased path pairs} by relaxing ``common internal vertex'' to ``common vertex'' in the definition of biased path pairs.
If there is a directed path matching $\mathcal{P} = \set{P_1, \dots, P_k}$ from $f ( \source{I} )$ to $g ( \sink{I} )$ without any weakly biased path pairs (not necessarily satisfying the path-set conditions),
we can easily construct a directed path matching $\mathcal{P'} = \set{P'_1, \dots, P'_k}$ from $\source{I}$ to $\sink{I}$ satisfying the path-set conditions: $P'_i$ is obtained from $P_i$ by appending $f^{-1}(s(P_i))$ before $P_i$ and $g^{-1}(t(P_i))$ after $P_i$.
Thus, in the following, it suffices to show that $\mathcal{P}$ can be constructed in $\bigO{|V|^2}$ time.

Fix an arbitrary vertex $r \in V$.
For each $v \in V$, we denote by $P_{r, v}$ the unique path between $r$ and $v$ in $\und{T}$.
Let $\overrightarrow{P_{r, v}}$ be a directed path obtained by orienting each edge of $P_{r, v}$ from $r$ to $v$.
We define $d_r(v) = \mathrm{fwd}(\overrightarrow{P_{r, v}}) - \mathrm{rev}(\overrightarrow{P_{r, v}})$, where $\mathrm{fwd}(\overrightarrow{P_{r, v}})$ and $\mathrm{rev}(\overrightarrow{P_{r, v}})$ denote the numbers of arcs in $T$ that $\overrightarrow{P_{r, v}}$ passes in the forward and reverse directions, respectively.

Using the values of $d_r(v)$ for $v \in V$, we iteratively construct a directed path matching $\mathcal{P} = \{P_1, \dots, P_k\}$ from $f(\source{I})$ to $g(\sink{I})$ as follows.
In the $i$-th step of the algorithm, we construct $P_i$.
Let $X_i \subseteq f(\source{I})$ and $Y_i \subseteq g(\sink{I})$ be the sets of vertices that are not chosen for the sources and sinks of $P_1, \dots, P_{i-1}$, respectively.
Throughout the execution of the algorithm, we keep an invariant that for each arc $e$, $w(e; X_i, Y_i) \geq 0$, which means that there is a directed path matching from $X_i$ to $Y_i$ by \cref{lem:chara_matching}.

Let $x \in X_i$ be a vertex with the smallest value of $d_r(x)$ among $X_i$.
Let $Y'_i \subseteq Y_i$ be the set of vertices to which there is a directed path from $x$ consisting of only arcs with $w(e; X_i, Y_i) > 0$.
Since there exists some directed path matching from $X_i$ to $Y_i$ by our invariant that $w(e, X_i, Y_i) \ge 0$ for $e \in A$, the set $Y'_i$ is not empty. 
Choose an arbitrary vertex $y \in Y'_i$ with the smallest value of $d_r(y)$.
We set $P_i$ to the unique directed path from $x$ to $y$ in $T$.
Since we use only arcs $e$ with $w(e; X_i, Y_i) > 0$ to construct $P_i$, we have $w(e; X_{i+1}, Y_{i+1}) \geq 0$, where $X_{i + 1} = X_i \setminus \{x\}$ and $Y_{i + 1} = Y_{i} \setminus \{y\}$, for all arcs and the invariant still holds.
We repeat this until $X_i = \emptyset$ (and hence $Y_i = \emptyset$).

Here, we show that the obtained directed path matching $\set{P_1, \dots, P_k}$ does not contain any weakly biased path pairs.
Assume that $P_i$ and $P_j$ ($i < j$) have a common vertex $v$.
By the choice of $x$ in the construction above, $d_r(s(P_i)) \le d_r(s(P_j))$ holds.
Thus, we have
\[
\mathrm{dist}(s(P_i), v) = -d_r(s(P_i)) + d_r(v) \ge - d_r(s(P_j)) + d_r(v) = \mathrm{dist}(s(P_j), v).
\]
As $v$ is a common vertex of $P_i$ and $P_j$, $t(P_j) \in Y'_i$ and thus $d_r(t(P_i)) \le d_r(t(P_j))$ holds
by the choice of $y$. Then, similarly we have
\[
\mathrm{dist}(v, t(P_i)) = -d_r(v) + d_r(t(P_i)) \le -d_r(v) + d_r(t(P_j)) = \mathrm{dist}(v, t(P_j)).
\]
Therefore, $P_i$ and $P_j$ do not form a weakly biased path pair.
We can compute each $P_i$ in $\bigO{|V|}$ time, and thus the whole running time (including computing $f$ and $g$) is $\bigO{|V|^2}$.

\section*{Acknowledgments}
This work is partially supported by JSPS KAKENHI Grant numbers JP18H04091, JP18K11168, JP18K11169, JP19K11814, JP19K20350, JP19J21000, JP20K19742, JP20K23323, JP20H05793, JP20H05795, JP21H03499, and JP21K11752, and JST, CREST Grant Number JPMJCR18K3, Japan.

\end{document}